\documentclass[twoside,leqno]{article}

\title{Fixed-Price Approximations in Bilateral Trade}
\author{Zi Yang Kang, Francisco Pernice, Jan Vondrák}
\date{\today}

\usepackage[letterpaper]{geometry}

\usepackage{ltexpprt}
\usepackage{hyperref}

\usepackage[utf8]{inputenc}
\usepackage{graphicx}
\usepackage{amsmath}
\usepackage{amsfonts}
\usepackage{amssymb}
\usepackage{mathrsfs}
\usepackage{float}
\usepackage{subcaption}
\usepackage{tikz}
\usepackage{framed,graphicx,xcolor}
\usepackage{array}
\usepackage[maxbibnames=99]{biblatex}
\usepackage{mathtools}
\usepackage{enumerate}
\usepackage{enumitem}
\usepackage{titling}

\newcolumntype{L}{>{$}l<{$}}
\newcolumntype{R}{>{$}R<{$}}

\definecolor{shadecolor}{gray}{0.9}

\begin{document}

\newtheorem{thm}{Theorem}
\newtheorem{lem}[thm]{Lemma}
\newtheorem{rem}[thm]{Remark}
\newtheorem{prop}[thm]{Proposition}
\newtheorem{cor}[thm]{Corollary}
\newtheorem{ex}[thm]{Example}
\newtheorem{defn}[thm]{Definition}
\newtheorem{problem}[thm]{Problem}

\newcommand{\norm}[1]{\left\|#1\right\|}
\newcommand{\OO}{\mathcal{O}}
\newcommand{\NN}{\mathcal{N}}
\newcommand{\E}{\mathbb{E}}
\newcommand{\infsum}{\sum_{n=1}^\infty}
\newcommand{\infinfsum}{\sum_{n=-\infty}^\infty}
\newcommand{\interior}{\text{int}\,}
\newcommand{\infsumi}{\sum_{i=1}^\infty}
\newcommand{\infsumj}{\sum_{j=1}^\infty}
\newcommand{\infsumk}{\sum_{k=1}^\infty}
\newcommand{\sumin}{\sum_{i=1}^n}
\newcommand{\sumiN}{\sum_{i=1}^N}
\newcommand{\sumkj}{\sum_{k=1}^j}
\newcommand{\sumkN}{\sum_{k=1}^N}
\newcommand{\sumkM}{\sum_{k=1}^M}
\newcommand{\1}{{\bf 1}}
\newcommand{\0}{{\bf 0}}
\newcommand{\rank}{{\bf rank}\;}
\newcommand{\ceil}[1]{\left \lceil{#1}\right \rceil }
\newcommand{\floor}[1]{\left \lfloor{#1}\right \rfloor }
\newcommand{\red}{\textcolor{red}}
\newcommand{\mini}{\text{minimize}\quad}
\newcommand{\maxi}{\text{maximize}\quad}
\newcommand{\Leq}{\preceq}
\newcommand{\Geq}{\succeq}
\newcommand{\Le}{\prec}
\newcommand{\Ge}{\succ}
\newcommand{\subjto}{\text{subject to}\quad}
\newcommand{\minst}[2]{\mini \;\,& #1 \\ \subjto &  #2}
\newcommand{\maxst}[2]{\maxi \;\,& #1 \\ \subjto &  #2}
\newcommand{\sumik}{\sum_{i=1}^k}
\newcommand{\tr}{\textbf{tr}\;}
\newcommand{\diag}{\textbf{diag}}
\newcommand{\sumim}{\sum_{i=1}^m}
\newcommand{\infsumm}{\sum_{m=1}^\infty}
\newcommand{\pr}{\mathbb{P}}
\newcommand{\sumjn}{\sum_{j=1}^n}
\newcommand{\sumnk}{\sum_{n=1}^k}
\newcommand{\sumjm}{\sum_{j=1}^m}
\newcommand{\inflim}{\lim_{n \rightarrow \infty}}
\newcommand{\lub}{\text{lub} \;}
\newcommand{\N}{\mathbb{N}}
\newcommand{\Z}{\mathbb{Z}}
\newcommand{\R}{\mathbb{R}}
\newcommand{\F}{\mathbb{F}}
\newcommand{\SSS}{\mathbb{S}}
\newcommand{\LL}{\mathcal{L}}
\newcommand{\ran}{\text{Ran}\,}
\newcommand{\Q}{\mathbb{Q}}
\newcommand{\I}{\mathcal{I}}
\newcommand{\st}{\;|\;}
\newcommand{\St}{\;\Bigg|\;}
\newcommand{\U}{\mathscr{U}}
\newcommand{\lt}{\langle}
\newcommand{\gt}{\rangle}
\newcommand{\Cant}{\mathcal{C}}
\newcommand{\FF}{\mathcal{F}}
\newcommand{\C}{\mathbb{C}}
\newcommand{\CC}{\mathcal{C}}
\newcommand{\RR}{\mathscr{R}}
\newcommand{\eps}{\varepsilon}
\newcommand{\vspan}[1]{\text{span}(#1)}
\newcommand{\intI}{\int_0^1}
\newcommand{\xn}[1][]{\{#1 x_n\}_{n=1}^\infty}
\newcommand{\DD}{\mathcal{D}}
\newcommand{\rightcomment}{\hfill \small \normalfont}
\newcommand{\yn}{\{y_k\}_{n=1}^\infty}
\newcommand{\pmed}{{p_{\text{median}}}}
\newcommand{\GammaFB}{\Gamma^{\text{FB}}}
\newcommand{\OPT}{\mathsf{OPT\text{-}W}}
\newcommand{\OPTGFT}{\mathsf{OPT\text{-}GFT}}
\newcommand{\OPTW}{\mathsf{OPT\text{-}W}}
\newcommand{\dd}{\mathrm{d}}
\newcommand{\g}{\gamma}

\newcommand{\GFT}{\mathsf{GFT}}
\newcommand{\LOSS}{\mathsf{LOSS}}
\newcommand{\calO}{\mathcal{O}}
\newcommand{\W}{\mathsf{W}}
\newcommand{\eqdist}{\stackrel{\mathcal{D}}{=}}
\newcommand{\AND}{\quad\text{and}\quad}
\newcommand{\e}{\eps}
\renewcommand{\(}{\left(}
\renewcommand{\)}{\right)}
\newcommand{\pp}{\partial}

\def\arraystretch{1.3}

\providecommand{\keywords}[1]
{
  \small	
  \textbf{\textit{Keywords---}} #1
}

\allowdisplaybreaks

\begin{titlingpage}

\maketitle

\begin{abstract}
We consider the bilateral trade problem, in which two agents trade a single indivisible item. It is known that the only dominant-strategy truthful mechanism is the fixed-price mechanism: given commonly known distributions of the buyer's value $B$ and the seller's value $S$, a price $p$ is offered to both agents and trade occurs if $S \leq p \leq B$. The objective is to maximize either expected welfare $\mathbb{E}[S + (B-S) \mathbf{1}_{S \leq p \leq B}]$ or expected gains from trade $\mathbb{E}[(B-S) \mathbf{1}_{S \leq p \leq B}]$.
     
We improve the approximation ratios for several welfare maximization variants of this problem. When the agents' distributions are identical, we show that the optimal approximation ratio for welfare is $\frac{2+\sqrt{2}}{4}$. With just one prior sample from the common distribution, we show that a $3/4$-approximation to welfare is achievable. When agents' distributions are not required to be identical, we show that a previously best-known $(1-1/e)$-approximation can be strictly improved, but $1-1/e$ is optimal if only the seller's distribution is known. 
\end{abstract}
    
\end{titlingpage}

\section{Introduction}

In this paper, we study the performance of fixed-price mechanisms in the canonical context of bilateral trade, in which a buyer and seller bargain over a single indivisible good. As is standard in the literature, agents have independent private values for the good, $B$ for the buyer and $S$ for the seller, and the mechanism designer is given some information about the distributions from which these values are drawn. We consider two basic settings in this paper: one where the distributions of the agents' values are identical (the ``{symmetric}'' case), and one where the distributions are arbitrary (the ``{general}'' or ``{asymmetric}'' case). 
We measure the performance of a mechanism by considering either (i) the ``gains from trade'': $B-S$ whenever trade occurs and $0$ otherwise, or (ii) the ``welfare'': $B$ if trade occurs and $S$ otherwise. Observe that the difference between welfare and gains from trade is simply the seller's value. 

The ``first-best'' optimum is considered to be $\max \{B,S\}$ for welfare, and $\max \{B-S, 0\}$ for gains from trade. In their seminal work, Myerson and Satterthwaite \cite{Myerson} showed that no budget-balanced, individually rational Bayesian incentive-compatible (BIC) mechanism attains the first-best optimum in general. They also presented a ``second-best'' mechanism which is budget-balanced, individually rational, BIC, and achieves the best possible gains from trade attainable by any such mechanism. However, the second-best mechanism is complicated and difficult to implement in practical settings, so a recent growing literature seeks simpler mechanisms that maintain provably near-optimal approximation guarantees. 

Motivated by the need for strategic simplicity, we focus on dominant-strategy incentive compatible (DSIC) mechanisms. A fundamental result of Hagerty and Rogerson \cite{hagerty} is that essentially every DSIC mechanism for bilateral trade has the following form: the mechanism designer chooses a (possibly random) price and offers trade at this price to the buyer and seller. Trade occurs if and only if the buyer's value exceeds the chosen price and the seller's value does not; this mechanism is DSIC. We refer to this as a {\em fixed-price} or \emph{posted-price} mechanism. Given this explicit characterization of the allowed mechanisms under our incentive compatibility criterion, other questions arise. How close to the efficient outcome can fixed-price mechanisms get? What information about the agents is required to set the optimal price? In this paper we give a comprehensive answer to these questions in the special case of symmetric bilateral trade, and partial answers in the general case.

\subsection{Prior Work}
Since the seminal work of Myerson and Satterthwaite \cite{Myerson}, a large literature has emerged that characterizes the optimal incentive-compatible mechanisms in various settings. When the strong budget balance condition is imposed, Hagerty and Rogerson \cite{hagerty} showed that dominant-strategy mechanisms for bilateral trade must essentially be fixed-price mechanisms, which we study in this paper for their strategic simplicity. Drexl and Kleiner \cite{drexlkleiner15} and Shao and Zhou \cite{shaozhou16} show that fixed-price mechanisms can be optimal even when the budget balance condition is relaxed to a no-deficit condition.

The two papers most closely related to ours are \cite{Dobzinski-Blumrosen} and \cite{efficientMarketsWithLimitedInformation}. In \cite{Dobzinski-Blumrosen}, Dobzinski and Blumrosen show that there is a randomized choice of a price based on the seller's distribution which provides a $(1-1/e)$-approximation to the optimal welfare (i.e., the expected welfare provided by the mechanism is at least $1-1/e$ times the optimal welfare). This approximation improves the bounds shown in previous work by \cite{blumrosendobzinski14} and \cite{colini-baldeschietal16}.  On the other hand, Dütting et al \cite{efficientMarketsWithLimitedInformation} show that one can get a $1/2$-approximation of welfare by just posting a sample from the seller distribution as a price (also in the asymmetric setting). Approximations to gains from trade have also been extensively studied in the literature \cite{best-of-both-worlds, approximating-gft-in-bilateral-trading, strongly-budget-balanced, McAfee, FixedPriceApproximability, Kang}. In particular, McAfee \cite{McAfee} showed that in the case of identical buyer's and seller's distributions, setting the price $p$ to be the {\em median} of the distribution achieves a $1/2$-approximation to optimal gains from trade. More recently, Kang and Vondrák \cite{Kang} showed that choosing the {\em mean} also achieves a $1/2$-approximation to optimal gains from trade, and this choice is in fact best possible among all fixed prices for a given distribution.   Even in more general two-sided markets (with more than two agents), Colini-Baldeschi et al \cite{approx-comb-auctions} and Brustle et al \cite{Brustle2017ApproximatingGF} have shown constant approximations of welfare and gains from trade, respectively. Along these lines, Cai et al \cite{Cai2021OnMG} have shown approximations even when various constraints on the transactions are enforced.

There have been various other results that show how first-best efficiency can be approximated when agents' values are drawn from restricted families of distributions (i.e., distributions satisfying certain regularity conditions), such as those due to Arnosti, Beck and Milgrom \cite{arnostietal16}, and Blumrosen and Mizrahi \cite{approximating-gft-in-bilateral-trading}. In this spirit, \cite{shaozhou16} demonstrates the optimality of the mean as a fixed-price mechanism in a bilateral trade setting for log-concave distributions that have an increasing hazard rate.\footnote{Another difference between the result in \cite{shaozhou16} and our Proposition~\ref{prop:price} is that they do not restrict attention to posted-price mechanisms {\em ex ante} as we do, but rather consider dominant-strategy  mechanisms with a no-deficit condition (as opposed to strong budget balance).} In contrast, our objective is to consider the most general class of distributions possible. As such, our results do not depend on distributional assumptions; all we require is that the distribution of agents' values has finite and nonzero mean.

\subsection{Our Results}

We extend the state of the art for bilateral trade in the following ways:
\begin{enumerate}
    \item We prove that posting a single sample of the common distribution as a price attains a $3/4$-approximation to the optimal welfare in the symmetric case (compared to a $1/2$-approximation by using a single sample in the asymmetric setting \cite{efficientMarketsWithLimitedInformation}). $3/4$ is also tight for this mechanism. 
    \item We show that in the symmetric setting, a fixed-price mechanism using the mean as a price attains a $\frac{2+\sqrt{2}}{4}$-approximation of the optimal welfare, which is the best possible among fixed-price mechanisms.
    \item For general (asymmetric) bilateral trade, with arbitrary agents' distributions, we prove that there is a $(1-1/e+\eps)$-approximation to optimal welfare for some constant $\eps>0$, improving the previously best-known $(1-1/e)$-approximation \cite{Dobzinski-Blumrosen}.
    \item We also show that the $1-1/e$-approximation cannot be improved with the knowledge of the seller's distribution only (which is what the mechanism of \cite{Dobzinski-Blumrosen} uses).
\end{enumerate}

\begin{table}[H]
\footnotesize
\centering
\begin{tabular}{|c|cc|cc|}
\hline
\textbf{Variant} & \multicolumn{2}{c|}{\textbf{Welfare approximation}} & \multicolumn{2}{c|}{\textbf{Gains from trade approximation}} \\ \hline
Symmetric, full knowledge  & {\color{magenta}$(2+\sqrt{2})/4$}* & & $1/2$*& \cite{Kang, McAfee} \\
\hline
Symmetric, 1 prior sample & {\color{magenta}$3/4$} & &  {  $1/2$}*
  & \cite{1/2-gft-one-sample}  \\
\hline
Asymmetric, full knowledge  & {\color{magenta} $1 - 1/e + \eps$} & & $0$*& \cite{Dobzinski-Blumrosen}\\
\hline
Asymmetric, 1 prior sample & $1/2$ & \cite{efficientMarketsWithLimitedInformation}&$0$*&  \cite{Dobzinski-Blumrosen} \\               
\hline
\end{tabular}
    \caption{A survey of fixed-price mechanism approximations for bilateral trade. Our contributions are highlighted in {\color{magenta}magenta}. Provably optimal results (within their variant) are marked with a *.
    \label{tab:survey}}
\end{table}

{\em Note:} A preliminary version of this paper incorrectly claimed a 1/2-approximation of gains from trade with one sample in the symmetric setting as a new result. Since then, we have been informed of a work by Babaioff, Goldner and Gonczarowski \cite{1/2-gft-one-sample} which contains a proof of this result. 
Since our proof is somewhat different, we include it in the appendix.

\section{Model}
Here we formally outline the model that was discussed above.

\paragraph{Bilateral Trade.} The bilateral trade problem is simple: two agents, a buyer and a seller, value an item that is held by the seller at $B$ and $S$ respectively. The specific realizations of $B$ and $S$ are unknown to us, but we assume that we have $B\sim F_B$ and $S\sim F_S$, with $B$ and $S$ independent, for distribution functions $F_S$ and $F_B$, about which we are given some information. 

\paragraph{Bilateral Trade Mechanisms.} A bilateral trade mechanism consists of an allocation function $A: \R\times\R\to \{0,1\}$, which takes in as input the reported valuations of the buyer and seller and outputs 1 if a transaction should occur or 0 otherwise, and a payment function $\Pi: \R\times \R \to \R$ which, if a trade occurs, determines the price at which the item should be transacted. 
By the revelation principle \cite{revelationPrinciple}, 
we are interested in \emph{incentive-compatible} mechanisms $(A,\Pi)$, where the agents are incentivized to report their true values. There are two standard notions of incentive compatibility that are considered in the literature: Bayesian incentive compatibility (BIC), and dominant-strategy incentive compatibility (DSIC). The former roughly says that reporting true values should be an optimal strategy for each agent, \emph{in expectation} over the possible behavior of the other agent. The latter notion, which is what we will consider in this paper, means that reporting true values is \emph{always} an optimal strategy for all agents, regardless of what the other agent does. 
Hence the notion that agents have incentive to act according to their true preferences is fully described by the DSIC property of a mechanism, which is formally
defined as
$$
\begin{cases}
\Pi(s,b) - s\cdot A(s,b) \geq \Pi(s',b) - s\cdot A(s',b), \\
b\cdot A(s,b) - \Pi(s,b) \geq  b\cdot A(s,b') -  \Pi(s',b')
\end{cases}\quad\text{for all }s,s',b,b' \in \R.
$$
That is, both agents are always better off reporting their true preferences. It was shown by \cite{hagerty} that such mechanisms are essentially \emph{fixed-price} mechanisms, where the payment function $\Pi$ is taken to be a (possibly random) single value $p$, only depending on the distribution from which the valuations come from (and not on the valuations themselves). Then the allocation function is given by $A(b,s) = \1_{s\leq p\leq b}$: trade occurs if and only if the seller values the item less than $p$, and the buyer values it more.

\paragraph{Gains From Trade and Welfare.} As discussed above, we consider two benchmarks in this paper: gains from trade and welfare. The optimal gains from trade, given a distribution function $F$ from which the valuations are drawn, is defined as
$$
\OPTGFT(F) = \E [\1_{B \geq S}(B-S)],
$$
while the gains from trade achieved by a particular, non-random posted price $p$ is given by
$$
\GFT(p,F) = \E [\1_{B \geq p \geq S}(B-S)].
$$
Similarly, the corresponding welfare measures are defined as
\begin{align*}
    \OPTW(F) &= \E[S] + \E[\1_{B \geq S}(B-S)] = \E[S] + \OPTGFT(F), \\
    \W(p,F) &= \E[S] + \E[\1_{B\geq p \geq S}(B-S)] = \E[S] + \GFT(p,F).
\end{align*}
It is often useful to have simplified integral formulas for these quantities. We now state two such formulas, both initially derived by McAfee \cite{McAfee}; for completeness, we offer derivations for them in the appendix.
\begin{align*}
    \OPTGFT(F) &=   \int_0^\infty F(x)(1-F(x))\ \dd x, \\
    \GFT(p, F) &= F(p)\int_p^\infty (1-F(x))\ \dd x + (1-F(p))\int_0^p F(x)\ \dd x.
\end{align*}

\section{Main Results}
We organize the results into three subsections. In the first subsection, we focus on the setting where only one sample from $F$ is given, and show a $3/4$-approximation of welfare.
In the second subsection, we work in the setting of full access to the common distribution $F,$ and show that posting the mean of $F$ as a price gives a $(2+\sqrt{2})/4$-approximation of welfare, which is the best-possible in the symmetric model. Finally, in the third subsection we return to the asymmetric setting, where we now have different distributions $F_S$ and $F_B$ for the buyer and seller, respectively. In this setting we give a $(1+1/e+0.0001)$-approximation of welfare, showing that the mechanism of \cite{blumrosendobzinski14} is not optimal; however, we show that $1-1/e$ is indeed optimal if one uses only the seller distribution to set the price.

\subsection{Single-Sample Approximation of Welfare}

As mentioned before, Babaioff, Goldner and Gonczarowski \cite{1/2-gft-one-sample} have shown that posting a single sample as a price in the symmetric setting achieves \emph{exactly} a 1/2-approximation of the expected optimal gains from trade, under the assumption that the distribution is atomless. This matches the best-possible achievable ratio for this model even if full information about the distribution is given. We offer a somewhat different proof of this fact in the appendix (Theorem \ref{GFT-1/2}). 

By similar arguments, we obtain a 3/4 approximation of optimal welfare. Recall the representations of the welfare measures from the previous section:
\begin{align*}
    \OPTW (F) &=  \mu + \OPTGFT(F), \\
    \W(F, p) &=\mu + \GFT(F,p),
\end{align*}
where $\mu := \E[S] = \E[B].$ We then get the following result.
\begin{thm}\label{welfare-3/4}
The symmetric bilateral trade mechanism which under a valuation distribution $F$ posts a price $p\sim F$ achieves a 3/4-approximation of the optimal welfare.
\end{thm}
\begin{proof}
Using that $\E_{p\sim F}[\GFT(F,p)] / \OPTGFT(F) =1/2,$ we can write
\begin{align*}
    \frac{\E_{p\sim F}[\W(F, p)]}{\OPTW(F)} &= \frac{\mu + \E_{p\sim F}[\GFT(F,p)] }{\mu + \OPTGFT(F)} \\
    &= 1 -
    \frac{1}{2}\frac{\OPTGFT(F)}{\mu + \OPTGFT(F)} \\
    &= 1 -
   \lim_{a\to\infty} \frac{1}{2}\frac{\int_0^a F(x)(1-F(x))\ \dd x}{\int_0^a xf(x)\ \dd x + \int_0^a F(x)(1-F(x))\ \dd x} \\
    &= 1 -
   \lim_{a\to\infty} \frac{1}{2}\frac{\int_0^a F(x)\ \dd x- \int_0^a F(x)^2\ \dd x}{ aF(a) - \int_0^a F(x)^2\ \dd x }.
\end{align*}
Note that we have $\int_0^a F(x)\ \dd x \leq a^{1/2}(\int_0^a F(x)^2\ \dd x)^{1/2}$ by the Cauchy--Schwarz inequality. Thus letting $t:=(\int_0^a F(x)^2\ \dd x)^{1/2},$ and noting that $0\leq t\leq a^{1/2}$ and $F(a)\leq 1,$ we get
\begin{align*}
    \frac{\int_0^a F(x)\ \dd x- \int_0^a F(x)^2\ \dd x}{ aF(a) - \int_0^a F(x)^2\ \dd x } &\leq \frac{a^{1/2}t - t^2}{aF(a) - t^2} \\
    &\leq\frac{1}{F(a)}\frac{t(a^{1/2} - t)}{(a^{1/2} + t)(a^{1/2} - t)} \\
    &=\frac{1}{F(a)}\frac{t}{a^{1/2}+t}.\\
    &\leq \frac{1}{F(a)}\frac{1}{2}\\
    &\xrightarrow{a\to\infty}1/2.
\end{align*}
Thus we have $\frac{\E_{p\sim F}[\W(F, p)]}{\OPTW(F)} \geq 1-\frac{1}{2}\cdot \frac{1}{2} = 3/4,$ as desired.
\end{proof}

The approximation lower bound of 3/4 from Theorem \ref{welfare-3/4} for our mechanism is tight, as we show next.
\begin{lem}
The 3/4-approximation bound of Theorem \ref{welfare-3/4} is tight for the mechanism using a sample $p \sim F$ as a price.
\end{lem}
\begin{proof}
Take the distribution functions $F_r(x) = \min\{x^r, 1\}$ parameterized by $r\in \R_+.$ From the proof of Theorem \ref{welfare-3/4}, we get
\begin{align*}
    \frac{\E_{p\sim F}[\W(F_r, p)]}{\OPTW(F_r)}  &=  1 - \frac{1}{2}\frac{\int_0^1 x^r(1-x^r)dx}{r\int_0^1 x^r dx + \int_0^1 x^r(1-x^r)dx} \\
    &=  1 - \frac{1}{2}\frac{\frac{1}{r+1} - \frac{1}{2r+1}}{\frac{r}{r+1}+ \frac{1}{r+1} - \frac{1}{2r+1}} = 1 - \frac12 \cdot \frac{r}{2r^2 + 2r}\\
    &\xrightarrow{r\to 0} 1- \frac{1}{2}\cdot \frac{1}{2} = \frac{3}{4}.
\end{align*}
\end{proof}

\subsection{Best-Possible Approximation of Welfare}
We begin by showing that the optimal posted price in the symmetric setting is given simply by the mean of the common distribution.

\begin{prop}\label{prop:price}
The welfare-maximizing mechanism for symmetric bilateral trade is given by posting the mean of the common distribution $F$:
\[p^*=\E[S]=\E[B].\]
\end{prop}

\begin{proof}
For any positive price $p>0$, we may rewrite the welfare function as
\begin{align*}
\W(p,F)
&=\E[S]+\E[(B-S)\cdot\1_{B \geq p\geq p}]\\
&=\E[S]+\E[B\cdot\1_{B \geq p}]\cdot F(p)-\E[S\cdot\1_{S\leq p}]\cdot\left[1-F(p)\right].
\end{align*}
Because both agents' values are drawn from the same distribution, 
\[\E[B\cdot\1_{B \geq p}]=\E[S\cdot\1_{S \geq p}]=\E[S]-\E[S\cdot\1_{S\leq p}].\]
Thus:
\begin{align*}
\W(p,F)
&=\E[S]+\E[B\cdot\1_{B \geq p}]\cdot F(p)-\E[S\cdot\1_{S\leq p}]\cdot\left[1-F(p)\right]\\
&=\E[S]\cdot\left[1+F(p)\right]-\E[S\cdot\1_{S\leq p}]\\
&=\E[S]\cdot\left[1+F(p)\right]-p\cdot F(p)+\int_0^p F(s)\ \dd s.
\end{align*}
We thus obtain that for $\delta \in [0,p^*]$,
\[
\W(p^*;F)-\W(p^*-\delta;F)=\left[F(p^*)-F(p^*-\delta)\right]\cdot\left(\E[S]-p^*\right)-\delta\cdot F(p^*-\delta)+\int_{p^*-\delta}^{p^*}F(x)\ \dd x,
\]
and for $\delta \in [0,\infty)$,
\[
\W(p^*;F)-\W(p^*+\delta;F)=\left[F(p^*)-F(p^*+\delta)\right]\cdot\left(\E[S]-p^*\right)+\delta\cdot F(p^*+\delta)-\int_{p^*}^{p^*+\delta}F(x)\ \dd x.
\]
In the first expression, since $F$ is non-decreasing,
\[\int_{p^*-\delta}^{p^*}F(x)\ \dd x\geq \int_{p^*-\delta}^{p^*}F(p^*-\delta)\ \dd x= \delta\cdot F(p^*-\delta).\]
Likewise, in the second expression, the same observation yields
\[-\int_{p^*}^{p^*+\delta} F(x)\ \dd x\geq -\int_{p^*}^{p^*+\delta} F(p^*+\delta)\ \dd x=-\delta\cdot F(p^*+\delta).\]
Therefore, with $p^*=\E[S]=\E[B]$, substitution yields
\[\begin{dcases}
\W(p^*;F)-\W(p^*-\delta;F)&\geq \left[F(p^*)-F(p^*-\delta)\right]\cdot\left(\E[S]-p^*\right)=0,\\
\W(p^*;F)-\W(p^*+\delta;F)&\geq \left[F(p^*)-F(p^*+\delta)\right]\cdot\left(\E[S]-p^*\right)=0.
\end{dcases}\]
Therefore the optimal fixed-price mechanism is given by $p^*=\E[S]$.
\end{proof}


Proposition~\ref{prop:price} shows that only the first moment of the agents' distribution is required by the mechanism designer to determine the optimal fixed price; all higher moments are inconsequential. Therefore, precise knowledge of the agents' distribution, other than the mean of the distribution, is irrelevant.

In general, the fixed-price mechanism $p^*=\E[S]$ is not uniquely optimal. Suppose, for example, that each agent's value is 0 with probability $1/2$ and 1 with probability $1/2$. Then any price $p\in(0,1)$ is optimal. Nonetheless, under weak regularity conditions on $F$, the mean is the uniquely optimal price. 

\begin{cor}\label{cor:unique}
Suppose that $F$ is differentiable in a neighborhood of its mean, so that its density $f=F'$ is positive in that neighborhood. Then the optimal symmetric bilateral trade fixed-price mechanism that sets the price as the mean of the distribution $F$ is uniquely optimal.
\end{cor}

The proof of Corollary~\ref{cor:unique} can be directly obtained from the proof of Proposition~\ref{prop:price} by noting that the inequalities must hold strictly under the stated assumptions; we thus omit the details. In particular, Corollary~\ref{cor:unique} applies under the usual assumption that $F$ is continuously differentiable with positive density.

Despite Corollary~\ref{cor:unique}, we proceed without additional regularity assumptions on $F$ to obtain a general analysis. 
Given the mean $\mu:=\E[S]$ of the agents' distribution, we can now determine the expected welfare that the optimal fixed-price mechanism achieves, $\W(\mu; F)$:
\begin{align*}
\W(\mu;F)
&=\mu\cdot\left[1+F(\mu)\right]-\E[S\cdot\1_{S\leq \mu}]\\
&=\mu+\(\mu-\E[S\,|\, S\leq \mu]\)\cdot F(\mu).
\end{align*}
An interesting observation here is that the maximum expected welfare achieved by {\em any} fixed-price mechanism depends on the agents' distribution $F$ through only three quantities: (i)~the mean of the distribution, $\mu$; (ii)~the quantile of the mean, $F(\mu)$; and (iii)~the mean of the distribution conditional on being no greater than $\mu$, $\E[S\,|\, S\leq \mu]$. Consequently, any modification to $F$ that preserves these three quantities will lead to the same maximum expected welfare.

This observation motivates the following approach to find the approximation ratio we achieve by posting the mean as a price which, as we showed, is the best-possible among DSIC mechanisms:
\begin{align}\label{eq:M}
    \inf_F \frac{\W(\mu_F;F)}{\OPTW(F)} = \inf_F \sup_{p \in \R_+} \frac{\W(p,F)}{\OPTW(F)}.\tag{A}
\end{align}
Define the subspace of probability distributions that fixes the three quantities from above:
\[\Delta_{L^1}(\R_+;\mu,\mu_1,\g):=\left\{F\in\Delta_{L^1}(\R_+):\E[S]=\mu,\,\E[S\,|\, S\leq \mu]=\mu_1,\,F(\mu)=\g\right\}.\]
To compute \eqref{eq:M}, we first compute
\begin{equation}\label{eq:M'}
\adjustlimits \inf_{F\in\Delta_{L^1}(\R_+;\mu,\mu_1,\g)}\sup_{p\in\R_+}\frac{\W(p,F)}{\OPTW(F)}.\tag{A'}\end{equation}
The original problem \eqref{eq:M} is thus equivalent to 
\[\inf_{\substack{\mu\,\mu_1\,\g\in\R_+\\ 0\leq\mu_1\leq\mu;\,\mu>0;\,0<\g\leq1}}\left\{\adjustlimits\inf_{F\in\Delta_{L^1}(\R_+;\mu,\mu_1,\g)}\sup_{p\in\R_+}\frac{\W(p,F)}{\OPTW(F)}\right\}.\]

The objective of this decoupling is to reduce the dimension of our optimization problem. Instead of minimizing over the infinite-dimensional space $\Delta_{L^1}(\R_+)$, our outer optimization problem is simplified to minimization over three variables. However, the inner optimization problem \eqref{eq:M'} still requires minimizing over the infinite-dimensional space $\Delta_{L^1}(\R_+;\mu,\mu_1,\g)$. To reduce the dimension of \eqref{eq:M'}, define the space of distributions {\em supported on at most 4 points} that fixes the three quantities:
\begin{align*}
    \Delta_{L^1}^{(4)}(\R_+;\mu,\mu_1,\g):=\{F\in\Delta_{L^1}(\R_+;\mu,\mu_1,\g): F(x)=&q_0\cdot\1_{x\geq0}+q_1\cdot\1_{x\geq x_1} \\
    &+q_2\cdot\1_{x\geq x_2}+q_3\cdot\1_{x\geq 1}\}.
\end{align*}
The probability masses are given by $0\leq q_0,q_1,q_2,q_3\leq 1$.


\begin{lem} \label{lem:dimensionality}
For any fixed $0\leq \mu_1\leq\mu\leq1$, $\mu>0$ and $0<\g\leq1$,
\[\adjustlimits \inf_{F\in\Delta_{L^1}(\R_+;\mu,\mu_1,\g)}\sup_{p\in\R_+}\frac{\W(p,F)}{\OPTW(F)}= \inf_{F\in\Delta_{L^1}^{(4)}(\R_+;\mu,\mu_1,\g)}\sup_{p\in\R_+}\frac{\W(p,F)}{\OPTW(F)}.\]
\end{lem}

We defer the proof of Lemma~\ref{lem:dimensionality} to Appendix~\ref{app:proofs} and sketch the argument here. Fix $F\in\Delta_{L^1}(\R_+;\mu,\mu_1,\g)$.  We may assume without loss of generality that $F$ is a finite support discrete distribution because any $F\in\Delta_{L^1}(\R_+;\mu,\mu_1,\g)$ can be approximated arbitrarily well by a corresponding distribution in $\Delta_{L^1}(\R_+;\mu,\mu_1,\g)$ supported on finitely many points. Moreover, by rescaling if necessary, we may also assume that $F(1)=1$.

To solve the inner problem \eqref{eq:M'}, our earlier observation shows that, given $\mu$, $\mu_1$ and $\g$, $\W(\mu;F)$ is constant for {\em any} $F\in\Delta_{L^1}(\R_+;\mu,\mu_1,\g)$. That is:
\[\adjustlimits \inf_{F\in\Delta_{L^1}(\R_+;\mu,\mu_1,\g)}\sup_{p\in\R_+}\frac{\W(p,F)}{\OPTW(F)}= \inf_{F\in\Delta_{L^1}(\R_+;\mu,\mu_1,\g)}\frac{\mu+\(\mu-\mu_1\)\cdot\g}{\OPTW(F)}.\]
 Thus \eqref{eq:M'} can be solved by {\em maximizing} $\OPTW(F)$ over all $F\in\Delta_{L^1}(\R_+;\mu,\mu_1,\g)$. Note that, from the definition,
\[\OPTW(F)=\E[\max\{S,B\}]=\E[S]+\frac12\,\E[|B-S|].\]
We apply the following procedure, which does not decrease $\OPTW(F)$:
\begin{itemize}
\item For any probability mass in $(0,\mu)$, split the mass into two equal masses. Move each mass in opposite directions, until one mass hits the boundary of the interval $[0,\mu]$.
\item For any probability mass in $(\mu,1)$, split the mass into two equal masses. Move each mass in opposite directions, until one mass hits the boundary of the interval $[\mu+\delta,1]$, for some sufficiently small $\delta>0$.\footnote{Note that this $\delta>0$ is required so that this operation does not change $F(\mu)=\g$.}
\end{itemize}

Each of these operations is mean-preserving in the intervals $[0,\mu]$ and $[\mu,1]$, so neither changes the quantities $\mu=\E[S]$, $\mu_1=\E[S\,|\, S\leq \mu]$ and $\g=F(\mu)$. Therefore, the operations preserve the maximum expected welfare achieved by any fixed-price mechanism, $\W(\mu;F)=\mu+\(\mu-\mu_1\)\cdot\g$. However, they increase $\E[|B-S|]$ in the respective intervals, and so the operations increase $\OPTW(F)$. Applying these operations repeatedly, \footnote{In the proof of Lemma~\ref{lem:dimensionality} in Appendix~\ref{app:proofs}, we consider only the limiting distribution $\tilde F$ obtained from recursive application of these operations; however, these operations motivate the construction of $\tilde F$ in the proof.} we obtain a 4-point distribution $\tilde F\in \Delta_{L^1}^{(4)}(\R_+;\mu,\mu_1,\g)$ such that
\[\frac{\W(p;\tilde F)}{\OPTW(\tilde F)}\leq \frac{\W(p,F)}{\OPTW(F)}.\]
Since $\Delta_{L^1}^{(4)}(\R_+;\mu,\mu_1,\g)\subset \Delta_{L^1}(\R_+;\mu,\mu_1,\g)$, this proves the result of Lemma~\ref{lem:dimensionality}. 


As a consequence of Lemma~\ref{lem:dimensionality}, we can restrict attention to 4-point distributions 
when solving the minimax problem \eqref{eq:M}. 
Such distributions $F\in \Delta_{L^1}^{(4)}(\R_+;\mu,\mu_1,\g)$ can be written as\footnote{The definition of $\Delta_{L^1}^{(4)}(\R_+;\mu,\mu_1,\g)$ is slightly more general as it allows for any $F(x)$ in the form \[F(x)=q_0\cdot\1_{x\geq0}+q_1\cdot\1_{x\geq x_1}+q_2\cdot\1_{x\geq x_2}+q_3\cdot\1_{x\geq1}.\] However, from the proof sketch of Lemma~\ref{lem:dimensionality} (and further justified in the proof given in Appendix~\ref{app:proofs}), we may narrow our attention to distributions $F$ where $x_1=\mu$ and $x_2=\mu+\delta$ for sufficiently small $\delta>0$.}
\[F(x)=q_0\cdot\1_{x\geq0}+q_1\cdot\1_{x\geq \mu}+q_2\cdot\1_{x\geq\mu+\delta}+q_3\cdot\1_{x\geq1}\quad\text{for some }\delta\in(0,1-\mu).\]
Here, the probability masses $q_0$, $q_1$, $q_2$ and $q_3$ satisfy $q_0+q_1+q_2+q_3=1$ and the three additional conditions:
\[\begin{dcases}
\E[S] =q_1\mu+q_2\(\mu+\delta\)+q_3 &=\mu,\\
\E[S\,|\, S\leq\mu] = \frac{q_1\mu}{q_0+q_1} &=\mu_1,\\
F(\mu) = q_0 + q_1 &=\g.\end{dcases}\]
Since $\mu$, $\mu_1$ and $\g$ are given, we can rewrite $F$ as follows:\footnote{Details of this computation can be found in the proof of Lemma~\ref{lem:dimensionality} given in Appendix~\ref{app:proofs}.}
\[F(x)=\g\(1-\frac{\mu_1}{\mu}\)\cdot\1_{x\geq0}+\frac{\mu_1\g}{\mu}\cdot\1_{x\geq\mu}+\frac{1-\g-\mu+\mu_1\g}{1-\mu-\delta}\cdot\1_{x\geq\mu+\delta}+\frac{\mu\g-\mu_1\g-\delta-\delta\g}{1-\mu-\delta}\cdot\1_{x\geq1}.\]

This parametrization of $F$ allows us to easily compute $\W(\mu;F)$ and $\OPTW(F)$. We have shown previously that $\W(\mu;F)$ depends only on $\mu$, $\mu_1$ and $\g$:
\[\W(\mu;F)=\mu+\(\mu-\mu_1\)\cdot\g.\]
On the other hand, $\OPTW(F)$ can be computed by observing that $B,S$ must be either: (i) both no greater than $\mu$; (ii) both greater than $\mu$; or (iii) on different sides of $\mu$. Thus:
\begin{align*}
\OPTW(F) 
&=\E[\max\{B,S\}]\\
&=\E[\max\{B,S\}\cdot\1_{B,S\leq\mu}]+\E[\max\{B,S\}\cdot\1_{B,S>\mu}]\\
&\qquad+\E[\max\{B,S\}\cdot\1_{\max\{B,S\}>\mu}\1_{\min\{B,S\}\leq\mu}].
\end{align*}
The first and third contributions can be computed separately:
\[\begin{dcases}
\E[\max\{B,S\}\cdot\1_{B,S\leq\mu}]=[F(\mu)]^2\cdot\mu\(1-\pr[B,S=0\,|\, B,S\leq\mu]\)&=\mu_1\g^2\(2-\frac{\mu_1}{\mu}\),\\
\E[\max\{B,S\}\cdot\1_{\max\{B,S\}>\mu}\1_{\min\{B,S\}\leq\mu}]=2\,[F(\mu)]\cdot\E[S\cdot\1_{S>\mu}]&=2\g\(\mu-\g\mu_1\).\end{dcases}\]
The second contribution can be bounded above independently of $\delta$:
\begin{align*}
\E[\max\{B,S\}\cdot\1_{B,S>\mu}]
&=[1-F(\mu)]^2\cdot\left\{1-\(1-\mu-\delta\)\cdot\pr[B=S=\mu+\delta\,|\, B,S>\mu]\right\}\\
&=\(1-\g\)^2-\frac{\(1-\g-\mu+\mu_1\g\)^2}{1-\mu-\delta}\\
&\leq\(1-\g\)^2-\frac{\(1-\g-\mu+\mu_1\g\)^2}{1-\mu}.
\end{align*}
Consequently, we obtain an upper bound for $\OPTW(F)$ that is independent of $\delta$:
\[\OPTW(F)\leq \mu_1\g^2\(2-\frac{\mu_1}{\mu}\)+2\g\(\mu-\g\mu_1\)+\(1-\g\)^2-\frac{\(1-\g-\mu+\mu_1\g\)^2}{1-\mu}.\]
Therefore, we have
\[\frac{\W(\mu;F)}{\OPTW(F)}\geq \frac{\mu+\(\mu-\mu_1\)\g}{\mu_1\g^2\(2-\frac{\mu_1}{\mu}\)+2\g\(\mu-\g\mu_1\)+\(1-\g\)^2-\frac{\(1-\g-\mu+\mu_1\g\)^2}{1-\mu}}.\]
This allows us to bound the value of the minimax problem \eqref{eq:M} from below:
\begin{align*}
\adjustlimits \inf_{F\in\Delta_{L^1}(\R_+)}\sup_{p\in\R_+}\frac{\W(p,F)}{\OPTW(F)}
&=\inf_{\substack{\mu\,\mu_1\,\g\in\R_+\\0\leq\mu_1\leq\mu;\,\mu,\g\in(0,1]}}\left\{\inf_{F\in\Delta_{L^1}(\R_+;\mu,\mu_1,\g)}\frac{\W(\mu;F)}{\OPTW(F)}\right\}\\
&\geq \inf_{\substack{\mu\,\mu_1\,\g\in\R_+\\0\leq\mu_1\leq\mu;\,\mu,\g\in(0,1]}}\frac{\mu+\(\mu-\mu_1\)\g}{\mu_1\g^2\(2-\frac{\mu_1}{\mu}\)+2\g\(\mu-\g\mu_1\)+\(1-\g\)^2-\frac{\(1-\g-\mu+\mu_1\g\)^2}{1-\mu}}.
\end{align*}

The lower bound in the above expression can be computed. Notably, optimization is carried out over a 3-dimensional space, and so elementary techniques apply. While we relegate the computational details to Appendix~\ref{app:proofs}, we state the result here:
\begin{lem}\label{lem:l-bound}
For symmetric bilateral trade, the minimax value of the game is bounded below:
\[\adjustlimits \inf_{F\in\Delta_{L^1}(\R_+)}\sup_{p\in\R_+}\frac{\W(p,F)}{\OPTW(F)}\geq\frac{2+\sqrt{2}}{4}.\]
\end{lem}

This bound is, in fact, tight. To prove this, we demonstrate that the bound can be attained by explicitly constructing a minimizing sequence of distributions $\{F_n\}_{n=1}^\infty$. Formally, we consider
\[F_n(x)=(\sqrt{2}-1)\left(1-\frac1n\right)\cdot\1_{x\geq0}+(2-\sqrt2)\cdot\1_{x\geq 1/n}+\frac1n\,(\sqrt{2}-1)\cdot\1_{x\geq1}.\]
The mean of the distribution $F_n$ is $\mu_n=1/n$. Moreover, the quantile of the mean is $\g_n=F_n(\mu_n)=1-(\sqrt{2}-1)/n$ and
\[\mu_{1,n}=\E[S\,|\, S\leq\mu_n]=\frac{2-\sqrt{2}}{n-(\sqrt{2}-1)}.\]
Thus
\[\W(\mu_n;F_n)=\mu_n+\(\mu_n-\mu_{1,n}\)\g_n=\frac{\sqrt{2}}{n}-\frac{\sqrt{2}-1}{n^2}.\]
On the other hand, we can compute that
\[\OPTW(F_n)=\frac{4\,(\sqrt{2}-1)}{n}-\frac{4\sqrt{2}-1}{n^2}.\]
Therefore, this sequence of distributions achieves the lower bound in the limit as $n\to\infty$:
\[\lim_{n\to\infty}\frac{\W(\mu_n;F_n)}{\OPTW(F_n)}=\frac{\sqrt{2}}{4\,(\sqrt{2}-1)}=\frac{2+\sqrt{2}}{4}.\]
This proves our main result:
\begin{thm}\label{thm:main}
For symmetric bilateral trade, the minimax value of the game is
\[\adjustlimits \inf_{F\in \Delta_{L^1}(\R_+)}\sup_{p\in\R_+}\frac{\W(p,F)}{\OPTW(F)}=\frac{2+\sqrt{2}}{4}.\]
That is, for any distribution $F$, the designer can always select a price that achieves at least $(2+\sqrt{2})/4\approx0.8536$ of the total expected welfare under the optimally efficient outcome.
\end{thm}

\subsection{Approximations of Welfare in the Asymmetric Model}

We now relax the requirement that both agents' values are drawn from the same distribution. Let the seller's (resp.~buyer's) value be drawn from $F_S$ (resp.~$F_B$). Accordingly, we modify the notation used previously. We denote the welfare function by
\[\W(p;F_S,F_B):=\E[S+\(B-S\)\cdot\1_{B \ge p\geq S}],\]
where the expectation is taken over $S\sim F_S$ and $B\sim F_B$. Likewise, we denote the expected welfare under the first-best outcome by
\[\OPT(F_S,F_B):=\E[S+\(B-S\)\cdot\1_{B \ge S}].\]
Therefore the best-possible posted-price approximation ratio is
\begin{equation}
\adjustlimits \inf_{F_S,F_B\in\Delta_{L^1}(\R_+)}\sup_{p\in\R_+}\frac{\W(p;F_S,F_B)}{\OPT(F_S,F_B)}.\tag{A-ASYM}
\label{eq:A-ASYM}
\end{equation}

As mentioned previously, \cite{Dobzinski-Blumrosen} shows that the value of the minimax problem \eqref{eq:A-ASYM} is bounded below by $1-1/e$. This is done through explicit construction of a randomized mechanism that achieves this bound in expectation. While this mechanism requires the designer to use detailed distributional information, it has the interesting property that the designer requires knowledge of only the seller's distribution and not the buyer's. Our aim is to exploit the knowledge of both distributions towards a provable improvement. First, we prove the following:

\begin{prop}\label{prop:3/4}
Given distributions $F_S$ and $F_B$ that satisfy $\E[(S-B)_+]=0$, the designer can always select a price that achieves at least $3/4$ of the optimal expected welfare. That is:
\[\inf_{\substack{F_S,F_B\in\Delta_{L^1}(\R_+)\\ \E[(S-B)_+]=0}}\sup_{p\in\R_+}\frac{\W(p;F_S,F_B)}{\OPT(F_S,F_B)} \geq \frac34.\]
\end{prop}
We defer the proof of this proposition and all other results in this subsection to the appendix. We interpret Proposition~\ref{prop:3/4} as a statement of how asymmetry between the agents can lead to considerable, albeit bounded, welfare loss. In Proposition~\ref{prop:3/4}, the asymmetry is ``worst possible'' in the sense that the seller's value is almost surely no greater than the buyer's value. Yet, ``worst possible'' asymmetry is not the same as the worst-case scenario for the minimax problem \eqref{eq:A-ASYM}. Indeed, the designer can find a price that separates all the seller's possible values from the buyer's possible values; because trade does not occur when the buyer's value is exactly equal to this price, inefficiency may result.\footnote{We emphasize that this result, along with our other results in this paper, is not an artifact of the tiebreak rule that we have adopted, which specifies that trade occurs if and only if $B \ge p\geq S$. For instance, if a different tiebreak rule were to be chosen such that trade occurs if and only if $B \ge p \ge S$, then a similar result obtains in the limit $\E[(S-B)_+]\to0$ (e.g., let $F_B(x)=(1-1/n)\cdot\1_{x\geq 1/n-1/n^2}+1/n\cdot\1_{x\geq1}$).}

When asymmetry is not as severe, we can improve the lower bound to the minimax problem \eqref{eq:A-ASYM} through the mechanism proposed by \cite{Dobzinski-Blumrosen}:

\begin{prop}[Theorem 4.1 in \cite{Dobzinski-Blumrosen}]\label{prop:bd}
For any given distributions $F_S$ and $F_B$,
\[\sup_{p\in\R_+}\frac{\W(p;F_S,F_B)}{\OPT(F_S,F_B)}\geq 1-\frac1e+\frac1e\cdot\E[\(S-B\)_+].\]
\end{prop}

In view of Propositions~\ref{prop:3/4} and \ref{prop:bd}, our strict improvement to the $1-1/e$ bound can be achieved heuristically as follows:
\begin{itemize}
\item If $\E[(S-B)_+]$ is sufficiently small, then we construct a mechanism that exploits the severe asymmetry in the agents' distributions to get a performance ratio close to $3/4$, as in Proposition~\ref{prop:3/4}.
\item If $\E[(S-B)_+]$ is large, then the mechanism of \cite{Dobzinski-Blumrosen} already ensures a performance ratio strictly higher than $1-1/e$.
\end{itemize}
Our improvement is documented by the following:\footnote{We qualify our result in Theorem~\ref{thm:small} by noting that we did not optimize for the strict constant-factor improvement of 0.0001, as our goal was to demonstrate that such an improvement exists. We do not believe that optimizing for this constant would yield a substantially higher improvement.}
\begin{thm}\label{thm:small}
Given distributions $F_S$ and $F_B$, the designer can always select a price that achieves at least $1-1/e+0.0001$ of the optimal expected welfare. That is:
\[\adjustlimits\inf_{F_S,F_B\in\Delta_{L^1}(\R_+)}\sup_{p\in\R_+}\frac{\W(p;F_S,F_B)}{\OPT(F_S,F_B)}\geq 1-\frac1e+0.0001.\]
\end{thm}


As a counterpoint to Theorem~\ref{thm:small} and motivated by the mechanism proposed by \cite{Dobzinski-Blumrosen}, we consider mechanisms that use only information about the seller's distribution. The question that we ask here is:  Is it possible to achieve a strictly better performance than $1-1/e$ of optimal welfare with any mechanism that uses {\em only} information about the seller's distribution $F_S$? The answer is no:
\begin{thm}\label{thm:quantile}
Given distributions $F_S$ and $F_B$, for any (possibly randomized) mechanism that uses only quantile distributional knowledge of the seller's distribution, the designer can achieve no better than $1-1/e$ of the optimal expected welfare.
\end{thm}

\section{Concluding Remarks}

Our results give a comprehensive description of the power of dominant-strategy truthful mechanisms for bilateral trade, in terms of two widely-studied objectives in mechanism design: welfare and gains from trade, together with the informational requirements needed to achieve those approximations. 

We would like to distinguish our paper from previous work as follows. In our view, the main contribution of our paper is to show that fixed-price mechanisms can be effective {\em particularly in markets with symmetric agents}, for reasons other than strategic simplicity for agents. Two important properties of fixed-price mechanisms in such markets is that they are informationally simple from the designer's perspective while being approximately efficient at the same time. We view robustness to different benchmarks (i.e., optimal welfare and optimal gains from trade) as an important criterion for ``approximate efficiency.'' In our view, previous approximability and inapproximability results in the asymmetric case, which crucially depend on benchmark adopted, constitute mainly a negative result for the asymmetric case. 

A fundamental question remains in the asymmetric case, which is whether a constant-factor approximation to optimal gains from trade is possible with any Bayesian incentive-compatible mechanism. We already know that such a mechanism cannot be incentive-compatible in dominant strategies, i.e.~it cannot be a fixed-price mechanism. There are other techniques to design Bayesian incentive-compatible mechanisms. In particular, the second-best mechanism is in some sense known, but its implicit description does not seem very useful for the derivation of approximation guarantees. It appears that a more specific Bayesian mechanism would be needed to resolve this question.


\printbibliography

\appendix

\section{Omitted Proofs}
\label{app:proofs}

We now derive the two simplified formulas for the optimal expected gains from trade and the expected gains from trade when posting price $p.$ We have
\begin{align*}
    \OPTGFT(F) &= \E[\1_{B \geq S}(B-S)] \\
    &= \int_0^\infty \int_0^b (b-s) f(s)f(b)\ \dd s \ \dd b \\
    &= \int_0^\infty \left[(b-s) F(s)\bigg|_{s=0}^b + \int_0^b F(s)\ \dd s\right]f(b)\ \dd b \\
    &= \lim_{B\to\infty}\int_0^B \int_0^b F(s)\ \dd s\cdot f(b)\ \dd b \\
    &= \lim_{B\to\infty}\left[\int_0^b F(s)\ \dd s\cdot F(b)\right]_{b=0}^B - \int_0^B F(b)^2 \ \dd b \\
    &= \lim_{B\to\infty }\int_0^B F(s)\ \dd s \cdot F(B) - \int_0^B F(b)^2\ \dd b\\
    &= \int_0^\infty F(x)(1-F(x))\ \dd x.
\end{align*}
Similarly for the expected gains from trade when posting price $p$, using integration by parts we get
\begin{align*}
    \GFT(p, F) &= \E[\1_{B\geq p \geq S}(B-S)] \\
    &= \E[\1_{B\geq p \geq S}(B-p)] + \E[\1_{B\geq p \geq S}(p-S)] \\
    &= \E[\1_{B\geq p}(B-p)]\pr(S \leq p) + \E[\1_{p \geq S}(p-S)] \pr(B\geq p)\\
    &= F(p)\int_p^\infty f(b)(b-p)\ \dd b + (1-F(p))\int_0^pf(s)(p-s)\ \dd s\\
    &= F(p)\int_p^\infty (1-F(x))\ \dd x + (1-F(p))\int_0^p F(x)\ \dd x.
\end{align*}
Using these formulas, we now offer an independent proof of the result of Babaioff, Goldner and Gonczarowski \cite{1/2-gft-one-sample} that posting a sample as a price in the symmetric model achieves a 1/2-approximation of welfare.

\begin{thm}\label{GFT-1/2}
The symmetric bilateral trade mechanism which under a valuation distribution $F$ posts a price $p\sim F$ achieves exactly 1/2 of the optimal gains from trade.
\end{thm}
\begin{proof}
We can write $\E_{p\sim F}[\GFT(p,F)]$ as
$$
\int_0^\infty \left[ F(p)\int_p^\infty (1-F(x))\ \dd x+ (1-F(p))\int_0^p F(x) \ \dd x\right]f(p)\ \dd p.
$$
We now define $\gamma_1 := \int_0^\infty f(p) F(p)\int_p^\infty (1-F(x))\ \dd x\ \dd p$ and $\gamma_2 := \int_0^\infty f(p)(1-F(p))\int_0^p F(x)\ \dd x\ \dd p$, which we will compute separately for simplicity, since we then have $\E_{p\sim F}[\GFT(p,F)] = \gamma_1 + \gamma_2.$ Let $\lambda := \int_0^\infty (1-F(x))\ \dd x.$ Now for the first term, we have
\begin{align*}
    \gamma_1 &= \int_0^\infty f(p) F(p)\int_p^\infty (1-F(x))\ \dd x\ \dd p \\
    &=  \int_0^\infty f(p) F(p) \left(\lambda - \int_0^p (1-F(x))\ \dd x\right)\ \dd p \\
    &=  \lambda \int_0^\infty f(p) F(p)\ \dd p -  \int_0^\infty f(p) F(p)\int_0^p (1-F(x))\ \dd x\ \dd p \\
    &=  \frac{\lambda}{2} -  \int_0^\infty f(p) F(p)\int_0^p (1-F(x))\ \dd x\ \dd p \\
    &=  \frac{\lambda}{2} -  \frac{1}{2}F(p)^2 \int_0^p (1-F(x))\ \dd x\ \Bigg|_0^\infty + \frac{1}{2}\int_0^\infty  F(p)^2(1-F(p))\ \dd p \\
    &=  \frac{1}{2}\int_0^\infty  F(p)^2(1-F(p))\ \dd p,
\end{align*}
where we used that 
$$
\int_0^\infty f(x)F(x)\ \dd x = \frac{1}{2}F(x)^2\ \bigg|_0^\infty = 1/2.
$$
Note$ \int_0^a f(x)(1-F(x))\ \dd x = F(a)(1-\frac{1}{2}F(a))$ by integrating by parts. Thus for the second term we can write
\begin{align*}
    \gamma_2 &= \lim_{a\to\infty} \int_0^a f(p)(1-F(p))\int_0^p F(x)\ \dd x\ \dd p \\
    &= \lim_{a\to\infty}\left[F(a)(1-\frac{1}{2}F(a))\int_0^aF(x)\ \dd x - \int_0^a F(p)^2(1-\frac{1}{2}F(p))\ \dd p\right].
\end{align*}
Now for $\theta := F(a)(1-\frac{1}{2}F(a))$ we have $\theta\to 1/2$ as $a\to\infty.$ Thus we have
\begin{align*}
    \gamma_2 &=\lim_{a\to\infty}\left[\theta\int_0^a F(x)\ \dd x - \int_0^a F(p)^2(1-\frac{1}{2}F(p))\ \dd p\right] \\
    &=\int_0^\infty\left[\frac{1}{2} F(p) - F(p)^2\left(1-\frac{1}{2}F(p)\right)\right]\ \dd p \\
    &=\frac{1}{2}\int_0^\infty F(p)(1-F(p))^2\ \dd p.
\end{align*}
Putting things together we get
\begin{align*}
    \E_{p\sim F}[\GFT(F;p)] &= \gamma_1 + \gamma_2 \\
    &= \frac{1}{2}\int_0^\infty  F(p)^2(1-F(p))\ \dd p + \frac{1}{2}\int_0^\infty F(p)(1-F(p))^2\ \dd p  \\
    &= \frac{1}{2}\int_0^\infty  F(p)(1-F(p))[F(p) + (1-F(p))]\ \dd p  \\
    &= \frac{1}{2}\int_0^\infty  F(p)(1-F(p))\ \dd p.
\end{align*}
Recalling that $\OPTGFT(F) = \int_0^\infty (1-F(x))F(x)\ \dd x,$ we get $\alpha=1/2,$ completing the proof.
\end{proof}

\begin{lem} 
For any fixed $0\leq \mu_1\leq\mu\leq1$, $\mu>0$ and $0<\g\leq1$,
\[\adjustlimits \inf_{F\in\Delta_{L^1}(\R_+;\mu,\mu_1,\g)}\sup_{p\in\R_+}\frac{\W(p;F)}{\OPTW(F)}= \inf_{F\in\Delta_{L^1}^{(4)}(\R_+;\mu,\mu_1,\g)}\sup_{p\in\R_+}\frac{\W(p;F)}{\OPTW(F)}.\]
\end{lem}
\begin{proof}
Fix $F\in\Delta_{L^1}(\R_+;\mu,\mu_1,\g)$. We may assume without loss of generality that $F$ is a finite discrete distribution, otherwise we may approximate $F$ arbitrarily well by such a distribution in $\Delta_{L^1}(\R_+;\mu,\mu_1,\g)$.\footnote{For instance, one may approximate $F$ by the finite discrete distributions corresponding to Riemann-sum approximations of $\E[S\,|\, S\leq\mu]$ (when restricted to $S\leq \mu$) and $\E[S\,|\, S>\mu]$ (when restricted to $S> \mu$).} 
By rescaling if necessary, we may also assume that $F(1)=1$. Let $F$ put positive probability mass only on each of the points $x_0,x_1,\ldots,x_N\in[0,1]$. 

Motivated by the procedure outlined in the proof sketch, we now construct $\tilde F\in\Delta_{L^1}^{(4)}(\R_+;\mu,\mu_1,\g)$ such that, for any $F\in\Delta_{L^1}(\R_+;\mu,\mu_1,\g)$,
\[\frac{\W(\mu;\tilde F)}{\OPTW(\tilde F)} \leq \frac{\W(\mu;F)}{\OPTW(F)}.\]
We write
\[
\tilde F(x)=q_0\cdot\1_{x\geq0}+q_1\cdot\1_{x\geq\mu}+q_2\cdot\1_{x\geq\mu+\delta}+q_3\cdot\1_{x\geq 1}.
\]
To construct $\tilde F$, we choose $q_0$, $q_1$, $q_2$, $q_3$ and $\delta$ as follows:
\begin{itemize}
\item $\delta$ is defined by
\[\delta=\min\{x_i-\mu: x_i>\mu,\,i\in\{1,2,\ldots,N\}\}.\]
\item $q_0$ and $q_1$ are defined by
\[\begin{dcases}
q_0+q_1 &=\g,\\
(q_0+q_1)\mu_1&=q_1\mu.\end{dcases}
\]
That is, $q_0=\g(1-\mu_1/\mu)$ and $q_1=\mu_1\g/\mu$.
\item If $\mu+\delta<1$, $q_2$ and $q_3$ are defined by 
\[\begin{dcases}
q_2+q_3 &=1-\g,\\
q_2(\mu+\delta)+q_3&=\mu-(q_0+q_1)\mu_1.\end{dcases}\]
That is, $q_2=(1-\g-\mu+\mu_1\g)/(1-\mu-\delta)$ and $q_3=(\mu\g-\mu_1\g-\delta+\delta\g)/(1-\mu-\delta)$.

If $\mu+\delta=1$, $q_2$ and $q_3$ are defined by $q_2=0$ and $q_3=1-\g$.
\end{itemize}

With $\tilde F$ defined above, we first verify that $\tilde F\in\Delta_{L^1}^{(4)}(\R_+;\mu,\mu_1,\g)$. Indeed:
\[\begin{dcases}
\tilde F(\mu)=q_0+q_1 &=\g,\\
\E_{\tilde F}[S] = q_1\mu+q_2(\mu+\delta)+q_3 &= \mu,\\
\E_{\tilde F}[S\,|\, S\leq \mu] = \frac{q_1\mu}{q_0+q_1} &=\mu_1.
\end{dcases}\]

Next, we show that, for any $F\in \Delta_{L^1}(\R_+;\mu,\mu_1,\g)$,
\[\OPTW(\tilde F)\geq \OPTW(F).\]
To do so, we compute that
\begin{align*}
\OPTW(F) 
&=\E_F[\max\{B,S\}]\\
&=\E_{F}[\max\{B,S\}\cdot\1_{B,S\leq\mu}]+\E_{F}[\max\{B,S\}\cdot\1_{B,S>\mu}]\\
&\qquad+\E_{F}[\max\{B,S\}\cdot\1_{\max\{B,S\}>\mu}\1_{\min\{B,S\}\leq\mu}].
\end{align*}
That is, $B,S$ must be either: (i) both no greater than $\mu$; (ii) both greater than $\mu$; or (iii) on different sides of $\mu$. For the last case, observe that:
\[\E_{F}[\max\{B,S\}\cdot\1_{\max\{B,S\}>\mu}\1_{\min\{B,S\}\leq\mu}]=\E_{\tilde F}[\max\{B,S\}\cdot\1_{\max\{B,S\}>\mu}\1_{\min\{B,S\}\leq\mu}].\]
Therefore, it suffices to show that $\E_{F}[\max\{B,S\}\cdot\1_{B,S\leq\mu}]\leq \E_{\tilde F}[\max\{B,S\}\cdot\1_{B,S\leq\mu}]$ and that $\E_{F}[\max\{B,S\}\cdot\1_{B,S>\mu}]\leq\E_{\tilde F}[\max\{B,S\}\cdot\1_{B,S>\mu}]$.

To show the former result, observe that
\[\begin{dcases}
\E_{F}[\max\{B,S\}\cdot\1_{B,S\leq\mu}]=2\int_0^\mu x\,F(x)\ \dd F(x)&=\mu\g^2-\int_0^\mu [F(x)]^2\ \dd x,\\
\E_{\tilde F}[\max\{B,S\}\cdot\1_{B,S\leq\mu}]=(q_0+q_1)^2\mu-q_0^2\mu&=\mu\g^2-q_0^2\mu.\end{dcases}\]
However, from the definition of $q_0$:
\[\int_0^\mu F(x)\ \dd x=\mu\g-\int_0^\mu x\ \dd F(x)=(\mu-\mu_1)\g=q_0\mu.\]
Therefore
\[q_0^2\mu=\frac1\mu\left[\int_0^\mu F(x)\ \dd x\right]^2\leq\int_0^\mu[F(x)]^2\ \dd x.\]
This establishes that $\E_{F}[\max\{B,S\}\cdot\1_{B,S\leq\mu}]\leq \E_{\tilde F}[\max\{B,S\}\cdot\1_{B,S\leq\mu}]$. The latter result follows by a similar argument.

We have thus proved that $\OPTW(\tilde F)\geq\OPTW(F)$. Moreover, since $\W(\mu;\tilde F)=\W(\mu; F)$, we have
\[\frac{\W(p;\tilde F)}{\OPTW(\tilde F)}\leq\frac{\W(p;F)}{\OPTW(F)}.\]
Since $\tilde F\in\Delta_{L^1}^{(4)}(\R_+;\mu,\mu_1,\g)\subset\Delta_{L^1}(\R_+;\mu,\mu_1,\g)$, hence this establishes:
\[\adjustlimits \inf_{F\in\Delta_{L^1}(\R_+;\mu,\mu_1,\g)}\sup_{p\in\R_+}\frac{\W(p;F)}{\OPTW(F)}= \inf_{F\in\Delta_{L^1}^{(4)}(\R_+;\mu,\mu_1,\g)}\sup_{p\in\R_+}\frac{\W(p;F)}{\OPTW(F)}.\]
\end{proof}

\begin{lem}
The value of the minimax problem \eqref{eq:M} is bounded below:
\[\adjustlimits \inf_{F\in\Delta_{L^1}(\R_+)}\sup_{p\in\R_+}\frac{\W(p;F)}{\OPTW(F)}\geq\frac{2+\sqrt{2}}{4}.\]
\end{lem}
\begin{proof}
From the argument given in the paper, it suffices to show that
\[\inf_{\substack{\mu\,\mu_1\,\g\in\R_+\\0\leq\mu_1\leq\mu;\,\mu,\g\in(0,1]}}\frac{\mu+(\mu-\mu_1)\g}{\mu_1\g^2(2-\frac{\mu_1}{\mu})+2\g(\mu-\g\mu_1)+(1-\g)^2-\frac{(1-\g-\mu+\mu_1\g)^2}{1-\mu}}=\frac{2+\sqrt{2}}{4}.\]
We begin with a change of variables. Define $x:=1-\mu_1/\mu\in[0,1]$; the minimization problem thus becomes
\[\inf_{\substack{x\in[0,1]\\\mu,\g\in(0,1]}}\frac{1+\g x}{1+2\g x-\frac{\g^2x^2}{1-\mu}}.\]
Observe that, for $x\in[0,1]$ and $\mu,\g\in(0,1]$,
\[\frac{\pp}{\pp \mu}\left[\frac{1+\g x}{1+2\g x-\frac{\g^2x^2}{1-\mu}}\right]=\frac{\g^2x^2(1+\g x)}{\left[1-\mu+\g x(2-2\mu-\g x)\right]^2}\geq 0.\]
That is, the objective function is non-decreasing in $\mu$. 

Thus the infimum of the objective function is achieved in the limit as $\mu\to0$:
\[\inf_{\substack{x\in[0,1]\\\mu,\g\in(0,1]}}\frac{1+\g x}{1+2\g x-\frac{\g^2x^2}{1-\mu}}=\inf_{\substack{x\in[0,1]\\\g\in(0,1]}}\frac{1+\g x}{1+\g x(2-\g x)}.\]
We now make another change of variables. Define $y:=\g x\in[0,1]$. The minimization problem thus becomes a single-dimensional minimization problem:
\[\inf_{\substack{x\in[0,1]\\\g\in(0,1]}}\frac{1+\g x}{1+\g x(2-\g x)}=\inf_{y\in[0,1]}\frac{1+y}{1+y(2-y)}=\frac{2+\sqrt{2}}{4}.\]
\end{proof}


\begin{prop}[Proposition \ref{prop:3/4} in the main text]
Given distributions $F_S$ and $F_B$ that satisfy $\E[(S-B)_+]=0$, the designer can always select a price that achieves at least $3/4$ of the total expected welfare under the first-best efficient outcome. That is:
\[\inf_{\substack{F_S,F_B\in\Delta_{L^1}(\R_+)\\ \E[(S-B)_+]=0}}\sup_{p\in\R_+}\frac{\W(p;F_S,F_B)}{\OPT(F_S,F_B)}\geq \frac34.\]
\end{prop}
\begin{proof}
The result to Proposition~\ref{prop:3/4} follows from Lemma~\ref{lem:lem-1}, which we prove below.
\end{proof}

\begin{lem} \label{lem:lem-1}
Fix $\alpha>0$. Given distributions $F_S$ and $F_B$ that satisfy $\E[(S-B)_+] \leq \alpha\cdot \OPT(F_S,F_B)$. 
Suppose there exist $p^+$ and $p^-$ such that:
\begin{enumerate}[label=(\roman*)]
\item $\pr[S> p^+]\leq \sqrt{10\alpha}$;
\item $\pr[B< p^-]\leq \sqrt{10\alpha}$; and
\item $0<p^+-p^-\leq\sqrt{10\alpha}$.
\end{enumerate}
Then the designer can set the price to be either $p^+$ or $p^-$ to achieve an expected welfare of at least 
\[\left(\frac34 - 2 \sqrt{10 \alpha}\right)\cdot \OPT(F_S,F_B).\]
\end{lem}

\begin{rem}
We note that Lemma~\ref{lem:lem-1} implies the result of Proposition~\ref{prop:3/4} as follows. Given distributions $F_S$ and $F_B$ that satisfy $\E[(S-B)_+]=0$, there exists a price $p^*$ for which $\pr[S\leq p^*]=\pr[B\geq p^*]=1$. Fix $\alpha>0$. Define
\[\begin{dcases}
p^+ := p^* + \frac12 \sqrt{10 \alpha} \cdot\OPT(F_S,F_B),\\
p^- := p^* - \frac12 \sqrt{10 \alpha} \cdot\OPT(F_S,F_B).
\end{dcases}\]
Observe that $\pr[S> p^+]=\pr[B< p^-]=0<\sqrt{10\alpha}$ and $p^+-p^-=\sqrt{10\alpha}$, so Lemma~\ref{lem:lem-2} applies for any $\alpha>0$. The result of Proposition~\ref{prop:3/4} is obtained in the limit as $\alpha\to0$. 
\end{rem}

\begin{proof}
For ease of notation, given a fixed-price mechanism $p$, denote the welfare loss relative to first-best efficiency by $\LOSS(p;F_S,F_B)$:
\[ \LOSS(p;F_S,F_B) := \OPT(F_S,F_B)-\W(p;F_S,F_B)=\E[(B-S)\cdot \1_{S<B \leq p}] + \E[(B-S)\cdot \1_{p<S<B}].\]
We begin by analyzing the welfare loss for $p^+$ and $p^-$. We have:
\[\begin{dcases}
\LOSS(p^+;F_S,F_B) & = \E[(B-S)\cdot \1_{S<B \leq p^+}] + \E[(B-S)\cdot \1_{p^+<S<B}] \\
 & \leq \E[B\cdot \1_{B<p^+} \1_{S<p^-}] + \E[(B-S)\cdot \1_{p^- \leq S < B \leq p^+}] + \E[B\cdot \1_{p^+<B} \1_{p^+<S}],\\
\LOSS(p^-;F_S,F_B) & = \E[(B-S)\cdot \1_{S<B \leq p^-}] + \E[(B-S)\cdot \1_{p^-<S<B}] \\
 & \leq \E[B\cdot \1_{B<p^-} \1_{S<p^-}] + \E[(B-S)\cdot \1_{p^- \leq S < B \leq p^+}] + \E[B\cdot \1_{B \geq p^+} \1_{S \geq p^-}].
 \end{dcases}\]

Consider the terms $\E[B\cdot \1_{B<p^+} \1_{S<p^-}]$ and $\E[B\cdot \1_{B \geq p^+} \1_{S \geq p^-}]$, which could potentially be large (i.e., close to the value of $\OPT(F_S,F_B)$).
However, they cannot be both large. Indeed, let $\beta:=\E[B\cdot \1_{B<p^+}]/\E[B]$ and $\sigma:=\pr[S<p^-]$. Because of the independence between $B$ and $S$,
we can write $\E[B\cdot \1_{B<p^+} \1_{S<p^-}] = \beta \sigma \cdot\E[B]$ and $\E[B \cdot\1_{B>p^+} \1_{S>p^-}]= (1-\beta) (1-\sigma)\cdot \E[B]$.
We distinguish between two cases: $\beta + \sigma \leq 1$ or $\beta + \sigma > 1$. In the first case, setting a price $p^+$ yields  
$$\E[B\cdot \1_{B<p^+} \1_{S<p^-}] \leq \beta \sigma \cdot\E[B] \leq \frac14 \cdot\E[B] \leq \frac14\cdot \OPT(F_S,F_B).$$
In the second case, setting a price $p^-$ yields
$$ \E[B\cdot \1_{B \geq p^+} \1_{S \geq p^-}] \leq (1-\beta) (1-\sigma)\cdot \E[B] \leq \frac14\cdot \E[B] \leq \frac14\cdot \OPT(F_S,F_B).$$
The remaining terms can be bounded by $\calO(\sqrt{\alpha}) \cdot\OPT(F_S,F_B)$ as follows: 
\[\begin{dcases}
 \E[(B-S) \cdot\1_{p^- \leq S < B \leq p^+}] \leq  p^+ - p^- &\leq \sqrt{10 \alpha} \cdot\OPT(F_S,F_B), \\
\E[B\cdot \1_{p^+<B} \1_{p^+<S}] \leq  \E[B] \cdot \pr[p^+ < S] &\leq \sqrt{10 \alpha} \cdot\OPT(F_S,F_B), \\
\E[B \cdot\1_{B<p^-} \1_{S<p^-}] \leq  \E[B] \cdot \pr[B<p^-] &\leq \sqrt{10 \alpha} \cdot\OPT(F_S,F_B).
\end{dcases}\]
Therefore, in either case, the welfare loss is bounded above by $\left(\frac14+2\sqrt{10\alpha}\right)\cdot \OPT(F_S,F_B)$. Thus
\[
\max_{p\in\{p^+,\,p^-\}}\W(p;F_S,F_B)\geq\left(\frac34-2\sqrt{10\alpha}\right)\cdot\OPT(F_S,F_B).
\]
\end{proof}

We now derive a sufficient condition for the hypotheses of Lemma~\ref{lem:lem-1} to hold:
\begin{lem}\label{lem:lem-2}
Fix $\alpha>0$. Given distributions $F_S$ and $F_B$ that satisfy $\E[(S-B)_+] \leq \alpha\cdot\OPT(F_S,F_B)$, suppose there exists $p^*$ such that 
\[
\pr[S \geq p^*] > \frac{1}{5}\AND\pr[B \leq p^*] > \frac{1}{5}.
\] 
Define
\[\begin{dcases}
p^+ := p^* + \frac12 \sqrt{10 \alpha} \cdot\OPT(F_S,F_B),\\
p^- := p^* - \frac12 \sqrt{10 \alpha} \cdot\OPT(F_S,F_B).
\end{dcases}\]
Then $\pr[S \geq p^+] \leq \sqrt{10 \alpha}$ and $\pr[B \leq p^-] \leq \sqrt{10 \alpha}$.
\end{lem}

\begin{proof}
Suppose to the contrary that $\pr[S \geq p^+] > \sqrt{10 \alpha}$.  Since $\pr[B \leq p^*] > 1/5$ by definition of $p^*$, we have \[\E[(S-B)_+] \geq (p^+-p^*) \cdot\pr[S \geq p^+]\cdot \pr[B \leq p^*]> \alpha \cdot\OPT(F_S,F_B),\quad\text{a contradiction}.\]
Similarly, if $\pr[B \leq p^-] > \sqrt{10 \alpha}$, then since $\pr[S \geq p^*] > 1/5$ by definition of $p^*$, we have \[\E[(S-B)_+] \geq (p^*-p^-)\cdot \pr[S \geq p^*]\cdot \pr[B \leq p^-] > \alpha \cdot\OPT(F_S,F_B),\quad\text{a contradiction}.\]
\end{proof}

Therefore, for small $\alpha>0$, Lemma~\ref{lem:lem-2} guarantees a strict improvement over the $1-1/e$ bound (i.e., that Lemma~\ref{lem:lem-1} applies) if there exists $p^*$ such that 
\[\pr[S \geq p^*] > \frac15\AND\pr[B \leq p^*] > \frac15.\]
What if such a $p^*$ does not exist? The following result ensures that we can still guarantee a strict improvement over the $1-1/e$ bound: 

\begin{lem}\label{lem:lem-3}
Given distributions $F_S$ and $F_B$, let $p^* := \inf \{p: \pr[S \geq p] \leq 1/5 \}$.
If $\pr[B \leq p^*] \leq 1/5$, then
\[\W(p^*;F_S,F_B) \geq \frac{17}{25}\cdot\OPT(F_S,F_B).\]
\end{lem}
\begin{proof}
Define $q := \pr[S \geq p^*] \leq 1/5$ and $r := \E[B\cdot\1_{B \leq p^*}] / \E[B]$. 
Observe that 
\[r \leq \pr[B \leq p^*] \leq \frac15.\]
As above, denote the welfare loss relative to first-best efficiency by $\LOSS(p;F_S,F_B)$:
\[ \LOSS(p;F_S,F_B) := \OPT(F_S,F_B)-\W(p;F_S,F_B)=\E[(B-S)\cdot \1_{S<B \leq p}] + \E[(B-S)\cdot \1_{p<S<B}].\]
We bound $\LOSS(p^*;F_S,F_B)$ from above as follows:
\begin{align*}
 \LOSS(p^*;F_S,F_B) 
 &\leq  \E[B\cdot \1_{B \leq p^*} \1_{S < p^*}] + \E[B\cdot \1_{B>p^*} \1_{S \geq p^*}]  \\
 &=  r \left(1-q\right) \cdot\E[B] + \left(1-r\right) q \cdot\E[B] \\
 &= \left[\frac12 - 2\left(\frac12-q\right)\left(\frac12-r\right) \right] \E[B]\\
 &\leq \left[\frac12 - 2\left(\frac12-\frac15 \right)^2 \right] \E[B] =  \frac{8}{25}\cdot \E[B]\leq\frac{8}{25}\cdot\OPT(F_S,F_B).
\end{align*}
Therefore
\[\W(p^*;F_S,F_B) =\OPT(F_S,F_B)-\LOSS(p^*;F_S,F_B)\geq \frac{17}{25}\cdot\OPT(F_S,F_B).\]
\end{proof}

\begin{prop}[Proposition \ref{prop:bd} in the main text; Theorem 4.1 in \cite{Dobzinski-Blumrosen}]
For any given distributions $F_S$ and $F_B$,
\[\sup_{p\in\R_+}\frac{\W(p;F_S,F_B)}{\OPT(F_S,F_B)}\geq 1-\frac1e+\frac1e\cdot\E[\left(S-B\right)_+].\]
\end{prop}
\begin{proof}
The proof is based almost entirely on the proof of Theorem~4.1 in \cite{Dobzinski-Blumrosen}. 
Since the mechanism depends only on the seller's distribution, we can work with a fixed buyer's value and then take an expectation over the buyer at the end.
Given $F_S$, we fix $b\in\R_+$ and consider the truncated seller's distribution $\tilde F_S$ (replacing all values above $b$ by $b$):
\[\tilde F_S(x)=F_S(x)\cdot\1_{x< b}+\1_{x\geq b}.\]
We denote by $\Phi_b$ the step function $\Phi_b(x) = \1_{x \geq b}$ (corresponding to a deterministic value of $b$).
For any distribution $G$ that depends only on $F_S$, we note that
\[\E_{p\sim G}[\W(p;\tilde F_S, \Phi_b)]=\E_{p\sim G}[\W(p;F_S,\Phi_b)] - \pr[S>b]\cdot\left(\E_{S\sim F_S}[S\,|\, S > b]-b\right), \]
because the last term is exactly the expected value that is lost by modifying the seller's distribution to $\tilde{F}_S$. (Note that the trade never happens when $S > b$, so the outcome in this case is always $S$.)
By the same argument,
\[\OPT(\tilde F_S, \Phi_b)=\OPT(F_S,\Phi_b)-\pr[S>b]\cdot\left(\E_{S\sim F_S}[S\,|\, S > b]-b\right). \]
Consider the distribution $G^*(x)=1+\log F_S(x)$, where $x\in[F_S^{-1}(1/e),F_S^{-1}(1)]$. \cite{Dobzinski-Blumrosen} show that 
\[\E_{p\sim G^*}[\W(p;\tilde F_S,\Phi_b)] \geq \left(1-\frac{1}{e}\right)\cdot\OPT(\tilde F_S,\Phi_b).\]
Consequently, substitution reveals that
\begin{align*}
\E_{p\sim G^*}[\W(p;F_S,\Phi_b)]
&\geq\left(1-\frac{1}{e}\right)\cdot\OPT(F_S,\Phi_b)+\frac{1}{e}\cdot\pr[S>b]\cdot\left(\E_{S\sim F_S}[S\,|\, S\geq b]-b\right) \\
&=\left(1-\frac{1}{e}\right)\cdot\OPT(F_S,\Phi_b)+\frac{1}{e}\cdot\E\left[(S-b)_+\right].
\end{align*}
Because $G^*$ depends only on $F_S$, it follows by linearity of expectation that 
\[\E_{p\sim G^*}[\W(p;F_S,F_B)]=\E_{b\sim F_B}\E_{p\sim G^*}[\W(p;F_S,\Phi_b)]\] and \[\OPT(F_S,F_B)=\E_{b\sim F_B}[\OPT(F_S,\Phi_b)].\] Taking expectations in the above yields
\begin{align*}
\E_{p\sim G^*}[\W(p;F_S,F_B)] &\geq \left(1-\frac{1}{e}\right)\cdot\OPT(F_S,F_B)+\frac{1}{e}\cdot\E\left[(S-B)_+\right]\\
&=\left(1-\frac{1-\alpha}{e}\right)\cdot\OPT(F_S,F_B).\end{align*}
\end{proof}

\begin{thm}[Theorem \ref{thm:small} in the main text]
Given distributions $F_S$ and $F_B$, the designer can always select a price that achieves at least $1-1/e+0.0001$ of the optimal expected welfare. That is:
\[\adjustlimits\inf_{F_S,F_B\in\Delta_{L^1}(\R_+)}\sup_{p\in\R_+}\frac{\W(p;F_S,F_B)}{\OPT(F_S,F_B)}\geq 1-\frac1e+0.0001.\]
\end{thm}
\begin{proof}
Let $\E[(S-B)_+] = \alpha\cdot \OPT$. If $\alpha \geq 0.0003$, then Proposition~\ref{prop:bd} yields 
\[\sup_{p\in\R_+}\frac{\W(p;F_S,F_B)}{\OPT(F_S,F_B)}\geq 1-\frac{1}{e}+\frac{\alpha}{e} > 1-\frac{1}{e}+0.0001.\]
If $0<\alpha < 0.0003$, then consider $p^*:=\inf\{p:\pr[S\geq p]\leq 1/5\}$. If $\pr[B\leq p^*]\leq1/5$, then Lemma~\ref{lem:lem-3} yields
\[\sup_{p\in\R_+}\frac{\W(p;F_S,F_B)}{\OPT(F_S,F_B)}\geq \frac{\W(p^*;F_S,F_B)}{\OPT(F_S,F_B)}\geq \frac{17}{25}.\]
Otherwise, if $\pr[B\leq p^*]>1/5$, then there exists some sufficiently small $\e>0$ such that $p_0=p^*-\e$ satisfies $\pr[S\geq p_0]>1/5$ and $\pr[B\leq p_0]>1/5$. Lemma~\ref{lem:lem-2} shows that this is a sufficient condition to satisfy the hypotheses of Lemma~\ref{lem:lem-1}, which implies:
\[\sup_{p\in\R_+}\frac{\W(p;F_S,F_B)}{\OPT(F_S,F_B)}\geq \frac34-2\sqrt{10\alpha}>\frac{16}{25}.\]
Finally, Proposition~\ref{prop:3/4} covers the case of $\alpha=0$.
\end{proof}

\begin{thm}[Theorem \ref{thm:quantile} in the main text]
Given distributions $F_S$ and $F_B$, for any (possibly non-deterministic) mechanism $G$ (where we let $G$ be a distribution in $[0,1]$ and we set the price $F_S^{-1}(G)$) that uses only quantile distributional knowledge of the seller's distribution, the designer can achieve no better than $1-1/e$ of the total expected welfare under the first-best efficient outcome. That is:
\[\adjustlimits\inf_{F_S,F_B\in\Delta_{L^1}(\R_+)}\sup_{G\in\Delta([0,1])}\frac{\W(G;F_S,F_B)}{\OPT(F_S,F_B)}= 1-\frac1e,\]
where $\W(G;F_S,F_B)$ denotes the expected welfare that the mechanism $G$ achieves.
\end{thm}
Before we present a formal proof, let us discuss a game-theoretic intuition behind this result. We can view the situation as a game between two players, the designer and nature. The designer tries to select a parameter $x$ to maximize efficiency, and nature tries to select a distribution $F_S$ to thwart the designer's goal. Given the choice of $x$ and $F_S$, the buyer's distribution $F_B$ is considered to be worst possible with respect to the designer's outcome. Our goal is to prove that there is a strategy of nature such that no strategy of the designer achieves an approximation factor better than $1-1/e$.

Due to von Neumann's theorem, there are optimal {\em mixed strategies} $\Xi$ and $\Phi$. It is important to keep in mind that these are randomized strategies: in the case of the designer, this means a random choice of $x$; in the case of nature, this means a random choice of $F_S$, (i.e., a probability distribution over cumulative distribution functions $F_S$, which is a more complicated object).

In order to simplify the game, let us make a few observations. Given $x$, $F_S$ and $F_B$, the expected outcome is given by taking an expectation over the buyer's value $b$ sampled from $F_B$ (because there is no dependency between $b$ and the choice of the price $p$ and the seller's value $s$). Therefore, we might as well assume that the buyer's value is deterministic, namely, the worst possible value $b$, given $x$ and $F_S$. Furthermore, for each choice of $F_S$ and $b$, the values can be rescaled so that $b=1$, without affecting the approximation ratio (i.e., the ratio of welfare relative to first-best efficiency). So we can assume without loss of generality that $b=1$.

Further, given that $b=1$, the seller's distribution can be truncated at $1$: any probability mass above $1$ can be moved to $1$.\footnote{There is a lemma making this argument in \cite{Dobzinski-Blumrosen} but since we are proving the opposite bound, this lemma is not formally needed here.}  This means that the first-best efficiency has value $\E[\max \{b,s\}] = 1$. 

Next, let us consider the strategy of nature. For a probability distribution $F_S$, if there is some mass between $(0,1)$, it only decreases the approximation ratio if we push this probability mass towards $0$ (the outcome possibly decreases, and the optimum is still $1$). It is important here that the probability mass is not concentrated on a single point -- relative comparisons between different possible values should still be non-trivial. However, we can assume for example that the probability mass below $1$ is uniform between $[0,\e]$, with density $y / \e$. Considering this, the only important parameter that governs the seller's distribution is the probability of $s$ between equal to $1$, $y = \pr[s < 1]$.

Hence, the game we are considering has pure strategies $x$ for the designer and $y$ for nature. Randomized strategies are distributions over $x$ and $y$. Given $x,y$ the payoff function for the mechanism (ignoring terms proportional to $\e$) is
$$ V(x,y) = (1-y) + x \cdot\1_{x<y}.$$
This reflects the fact that with probability $1-y$, the seller's value is $1$, in which case the outcome is certainly $1$ (since the buyer's value is also $1$).
Otherwise, the seller's value is close to $0$; then the trade occurs exactly when $s<p$ and $p<1$, and the outcome in that case is $1$; otherwise close to $0$. The event $x<y$ is equivalent to the fact that $p<1$, because $y = F_S(1)$. Given that $x<y$, the probability that $s<p$ is exactly $x$, because $p = F_S^{-1}(x)$. Therefore, $x \cdot\1_{x<y}$ is the contribution to the expected outcome in case the seller's value is below $1$.

We now derive the optimal mixed strategies. Let us assume that nature's strategy is given by a probability density function $\rho(y)$.
Then for a given (pure) designer's strategy $x$, nature's expected payoff is
\begin{equation}
\label{eq:1}
 \E[V(x,y)] = \E[(1-y) + x \cdot\1_{x<y}] = \int_0^x (1-y) \rho(y) \ \dd y + \int_x^1 (1-y+x) \rho(y)\ \dd y.\tag{\dag}
\end{equation}

We posit that for an optimal nature strategy $\rho(y)$, this quantity should be the same for every $x$ in the support of the optimal mechanism strategy. If not, then the designer's strategy could be modified to achieve a better outcome, by picking the $x$ maximizing the quantity above. We are trying to prove that the designer's strategy is defined by $g(x) = 1/x$ for $x \in [1/e,1]$. Hence, let us assume that the quantity in \eqref{eq:1} is constant for $x \in [1/e,1]$. By differentiating \eqref{eq:1} with respect to $x$, we obtain (for $x \in [1/e,1]$),
\begin{equation*}
\int_x^1 \rho(y)\ \dd y - x\, \rho(x) = 0.
\end{equation*}
Note that for $x = 1-\e$, we obtain 
\[\int_{1-\e}^{1} \rho(y)\ \dd y = (1-\e)\, \rho(1-\e).\] 
This is a somewhat paradoxical conclusion.
What this actually means is that the probability distribution cannot be fully defined by a density function; there must be a discrete probability mass at $x=1$,
which is equal to the density just below $1$.\footnote{These arguments are not  rigorous, but in any case we are just trying to guess the optimal form of nature's strategy.}
Differentiating one more time, we get
\[ -2 \rho(x) - x\, \rho'(x) = 0.\]
This differential equation is easy to solve: the solution is $\rho(y) = C / y^2$ for $y \in (1/e,1)$.
There should also be a discrete probability mass at $y=1$ equal to $C$. The normalization condition implies that $C = 1/e$. 
To complete the proof, we just have to show that there is no strategy of the designer that achieves a factor better than $1-1/e$ against this nature's strategy.

\begin{proof}
Motivated by the discussion above, we consider the following strategy of nature:
\begin{itemize}
\item With probability $1/e$, set $y = 1$.
\item With probability $1-1/e$, sample $y \in [1/e,1]$ with density $\rho(y) = 1/(e y^2)$.
\end{itemize}
Given $y$, nature's value $s$ is distributed as follows:
\begin{itemize}
\item With probability $y$, $s \in [0,\e]$ uniformly at random.
\item With probability $1-y$, $s = 1$.
\end{itemize}
We claim that for any strategy of the designer, the approximation ratio is at most $1-1/e$ against this nature's strategy.
Since mixed strategies are convex combinations of pure strategies, it is enough to consider pure strategies $x \in [0,1]$.

As we argued above, given $x$ and $y$, the expected value of the game, up to $\calO(\e)$ terms, is
$$ V(x,y) = (1-y) + x \cdot\1_{x<y}.$$
\noindent We have the following cases:
\begin{itemize}
\item $x \in [0,1/e]$. In this range, we certainly have $x < y$ (because $y$ is always at least $1/e$).
Thus, $V(x,y) = 1-y+x$. In expectation over $y$, this quantity is
\begin{align*}
\E[V(x,y)] &= \E[1-y+x] = \int_{1/e}^{1} (1-y+x)\cdot \frac{1}{ey^2}\ \dd y + \frac{1}{e} \cdot x\\
&=\left[-\frac{1+x}{ey} - \frac{1}{e} \ln y \right]_{y=1/e}^{1} + \frac{1}{e} \cdot x
  = 1 - \frac{2}{e} + x.
  \end{align*}
Since $x \leq 1/e$, we have $\E[V(x,y)] \leq 1-1/e$.

\item $x \in [1/e,1)$. In this range, we have $y \leq x$ or $y > x$ depending of the value of $y$. In the first case, the value of the game is $1-y$ and in the second case it is $1-y+x$. Thus we compute the expected value as follows:
\begin{align*}
\E[V(x,y)] &= \int_{1/e}^{1} (1-y)\cdot \frac{1}{e y^2}\ \dd y + \int_{x}^{1} x\cdot  \frac{1}{e y^2}\ \dd y + \frac{1}{e} \cdot x\\
&= \left[ -\frac{1}{ey} - \frac{1}{e} \ln y \right]_{y=1/e}^{1} + \left[ -\frac{x}{ey} \right]_{y=x}^{1} + \frac{1}{e} \cdot x\\
&=  \left(1 - \frac{2}{e} \right) + \left(\frac{1}{e} - \frac{x}{e} \right) + \frac{x}{e} = 1 - \frac{1}{e}.
\end{align*}

\item $x = 1$. Then we get the same expressions as above, except for the term $\frac{1}{e} \cdot x$, since the case of $y=1$ does not contribute anything.
Therefore, again the value is at most $1-1/e$.
\end{itemize}

We conclude that there no strategy of the designer that achieves an expected welfare of more than $1-1/e$ (neglecting $\calO(\e)$ terms).
The first-best efficiency is $1$ (since the buyer's value is always $1$) and hence the approximation factor cannot be better than $1-1/e$.
\end{proof}

\end{document}